\documentclass[11pt]{article}

\usepackage[utf8]{inputenc}
\usepackage[T1]{fontenc}
\usepackage[english]{babel}

\usepackage{microtype}

\usepackage{fullpage}
\usepackage{multicol}

\usepackage{amsmath,amssymb,amsfonts,amsthm}
\usepackage{mathabx}

\newtheorem{theorem}{Theorem}
\newtheorem{corollary}{Corollary}
\newtheorem{proposition}{Proposition}
\newtheorem{conjecture}{Conjecture}
\theoremstyle{definition}
\newtheorem{definition}{Definition}
\theoremstyle{remark}
\newtheorem{attempt}{Attempt}

\newcommand{\eps}{\varepsilon}
\newcommand{\Ot}[1]{\tilde{O}({#1})}

\def\polylog{\operatorname{polylog}}
\def\poly{\operatorname{poly}}

\usepackage[pdftex]{graphicx}
\usepackage[usenames,dvipsnames,table]{xcolor}
\usepackage{booktabs}

\usepackage{todonotes}

\usepackage{authblk}

\title{Popular Conjectures as a Barrier for \\ Dynamic Planar Graph Algorithms}

\author[1]{Amir Abboud\thanks{ Part of the work was performed while visiting the Simons Institute for the Theory of Computing, Berkeley, CA. Research partially supported by Virginia Vassilevska Williams's NSF Grants CCF-1417238 and CCF-1514339, and BSF Grant BSF:2012338. } }
\author[2]{Søren Dahlgaard\thanks{Part of this work was done while the
        author was visiting Stanford University. Research partly supported by
        Mikkel Thorup's Advanced Grant DFF-0602-02499B from the Danish Council for
        Independent Research under the Sapere Aude research career programme.}}
\affil[1]{Stanford University\\\texttt{abboud@cs.stanford.edu}}
\affil[2]{University of Copenhagen\\\texttt{soerend@di.ku.dk}}
\date{}

\begin{document}
\setcounter{page}{0}
\maketitle
\begin{abstract}
    The \emph{dynamic shortest paths} problem on planar graphs asks us to
    preprocess a planar graph $G$ such that we may support insertions and
    deletions of edges in $G$ as well as distance queries between any two nodes
    $u,v$ subject to the constraint that the graph remains planar at all times.
    This problem has been extensively studied in both the theory and
    experimental communities over the past decades and gets solved millions of
    times every day by companies like Google, Microsoft, and Uber. 
    The best known algorithm performs queries and updates in $\Ot{n^{2/3}}$ time, based on ideas
    of a seminal paper by Fakcharoenphol and Rao [FOCS'01]. 
    A $(1+\eps)$-approximation algorithm of Abraham \emph{et
    al.} [STOC'12] performs updates and queries in $\Ot{\sqrt{n}}$ time. 
    An algorithm with a more practical $O(\poly\log{n})$ runtime would be a major breakthrough.
    However, such runtimes are only known for a
    $(1+\eps)$-approximation in a model where only restricted weight updates
    are allowed due to Abraham \emph{et al.} [SODA'16], or for easier problems like connectivity. 

    In this paper, we follow a recent and very active line of work on showing
    lower bounds for polynomial time problems based on popular conjectures, obtaining the first such results for natural problems in \emph{planar graphs}. 
    Such results were previously out of reach due to the highly non-planar nature of known reductions and the impossibility of ``planarizing gadgets''.
    We introduce a new framework which is inspired by techniques from the literatures on distance labelling schemes and on parameterized complexity.
    

     Using our framework, we show that no algorithm for dynamic shortest paths or maximum weight bipartite
    matching in planar graphs can support both updates and queries in amortized $O(n^{\frac{1}{2}-\varepsilon})$ time, for $\eps>0$, unless the classical all-pairs-shortest-paths problem can be solved in truly subcubic time, which is widely believed to be impossible.
  We extend these results to obtain strong lower bounds for other related problems as well as for possible trade-offs between query and update time.
       Interestingly, our lower bounds hold even in
    very restrictive models where only weight updates are allowed.


\end{abstract}

\thispagestyle{empty}

\newpage
\setcounter{page}{1}
\section{Introduction}

The \emph{dynamic shortest paths} problem on \emph{planar graphs} is to preprocess a planar graph $G$, e.g. the national road network, so that we are able to efficiently support the following two operations:

\begin{itemize}

    \item At any point, we might \emph{insert} or \emph{remove} 
        an edge $(u,v)$ in $G$, e.g. in case a road gets congested due to an
        accident. Such updates are subjected to the constraint that the
        planarity of the graph is not violated. We may also consider another
        natural variant in which we are only allowed to update the weights of
        existing edges.
\item We want to be able to quickly answer \emph{queries} that ask for the
    length of the shortest path between two given nodes $u$ and $v$, in the most current graph $G$.

\end{itemize}

This is a problem that gets solved millions of times every day by companies
like Google, Microsoft, and Uber on graphs such as road networks with many
millions of nodes. It is thus a very important question in both theory and
practice whether there exists data structures that can perform updates and
(especially) queries on graphs with $n$ nodes in polylogarithmic or even
$n^{o(1)}$ time.

Shortest paths problems on planar graphs provide an ideal combination of
mathematical simplicity and elegance with faithful modeling of realistic
applications of major industrial interest.
The literature on the topic is too massive for us to survey in this paper: the current draft of the book ``Optimization Problems in Planar Graphs'' by Klein and Mozes \cite{KMbook} dedicates \emph{four chapters} to the algorithmic techniques for shortest paths by the theory community.
While near-optimal algorithms are known for most variants of shortest paths on static planar graphs, the \emph{dynamic} setting has proven much more challenging.

Since an $s,t$-shortest path in a planar graph can be found in near-linear time
(linear time for non-negative weights) \cite{HenzingerKS97,FakR06}, there is a naïve algorithm for the dynamic problem that spends $\tilde{O}(n)$ time on queries.
After progress on other related problems on dynamic planar graphs \cite{Fred85,
GalilI91,DjidPZ91,GalilIS92, Sub93,KleinS98, HenzingerKS97}, the first
sublinear bound was obtained in the seminal paper of Fakcharoenphol and Rao
\cite{FakR06}, which introduced new techniques that led to major results
for other problems like Max Flow (even on static graphs) \cite{BorKMNW11,LackNSW12}. The amortized time
per operation was $O(n^{2/3} \log^{7/3}n)$ and $O(n^{4/5} \log^{13/5}n)$ if
negative edges are allowed, and follow up works of Klein~\cite{Klein05},
Italiano \emph{et al.} \cite{ItalianoNSW11}, and Kaplan \emph{et al.}
\cite{KaplanMNS12} reduced the runtime to $O(n^{2/3} \log^{5/3}n)$ (even
allowing negative weights), and most recently,  Gawrychowski and Karczmarz
\cite{GawrK16} reduced it further to $O(n^{2/3}
\frac{\log^{5/3}n}{\log^{4/3}\log{n}})$. In fact these algorithms give a
trade-off on the update and query time of $\Ot{n/\sqrt{r}}$ and
$\Ot{r}$, for all $r$.
The problem has also been extensively studied from an engineering viewpoint on real-world transportation networks (see \cite{DellingSSW09} for a survey).
State of the art algorithms \cite{Bauer06,DellingGPW11,DellingW07,GeisbergerSSV12,SchultesS07} are able to exploit further structure of road networks (beyond planarity) and process updates to networks with tens of millions of nodes in milliseconds.

In a recent SODA'16 paper, Abraham \emph{et al.} \cite{AbrahamCDGW16} study
worst case bounds under a restricted but realistic model of dynamic updates in
which a base graph $G$ is given and one is allowed to perform only weight
updates subject to the following constraint: For any updated graph $G'$ it must
hold that $d_G(u,v) \le d_{G'}(u,v) \le M\cdot d_G(u,v)$ for all $u,v$ and some
parameter $M$. (Note that this will hold if, for example, the weight of each
edge only changes to within a factor of $M$.)
In this model, the authors obtain a $(1+\eps)$-approximation algorithm that maintains updates in $O(\poly\log{n} \cdot M^4/\eps^3)$ time.
Without this restriction, the best known $(1+\eps)$-approximation algorithms use $\tilde{O}(\sqrt{n})$ updates \cite{KleinS98,AbrahamCG12}.
Thus, when $M$ is small, this model allows for a major improvement over the
above results which require \emph{polynomial time} updates.
But is it enough to allow for \emph{exact} algorithms with \emph{subpolynomial} updates?
Such a result would explain the impressive experimental performance of state of the art algorithms.

On the negative side, Eppstein showed that $\Omega(\log n / \log\log n)$ time is required in the cell probe model \cite{Eppstein97} for planar connectivity (and therefore also shortest path).
However, an unconditional $\log^{\omega(1)}{n}$ lower bound is far beyond the scope of current techniques (see \cite{CliffordGL15}).
In recent years, much stronger lower bounds were obtained for dynamic problems under certain popular conjectures \cite{RodittyZ11,Patrascu10,AbboudV14,KopelowitzPP16,HenzingerKNS15,AbboudVY15,Dahlgaard16}.
For example, Roditty and Zwick \cite{RodittyZ11} proved an $n^{2-o(1)}$ lower
bound for dynamic single source shortest paths in general graphs under the
following conjecture.
\begin{conjecture}[APSP Conjecture]\label{conj:apsp}
    There exists no algorithm for solving the all pairs shortest paths (APSP)
    problem in general weighted (static) graphs in time $O(n^{3-\eps})$ for any
    $\eps > 0$.
\end{conjecture}

However, the reductions used in these results produce graphs that are
fundamentally non-planar, such as dense graphs on three layers, and
popular 
approaches for making them planar, e.g. by replacing each edge crossing with a small ``planarizing gadget'', are provably impossible (this was recently shown for matching \cite{GurjarKMST12} and is easier to show for problems like reachability and shortest paths).
Due to this and other challenges 
no (conditional) polynomial lower bounds were known for any natural problem on (static or dynamic) planar graphs.

On a more general note, an important direction for future research on the fine-grained complexity of polynomial time problems (a.k.a. Hardness in P) is to understand the complexity of fundamental problems on restricted but realistic classes of inputs.
A Recent result along these lines is the observation that the $n^{2-o(1)}$ lower bound for computing the diameter of a sparse graph \cite{RodittyV13} holds even when the treewidth of the graph is $O(\log{n})$ \cite{AbboudVW16}.
In this paper, we take a substantial step in this direction, proving the first strong (conditional) lower bounds for natural problems on planar graphs.

\subsection{Our Results}
We present the first conditional lower bounds for natural problems on
planar graphs using a new framework based on several ideas for conditional
lower bounds on dynamic graphs combined with ideas from parameterized
complexity \cite{Marx07,Marx13} and labeling schemes \cite{GavoillePPR04}. We
believe that this framework is of general interest and might lead to more
interesting results for planar graphs. Our framework shows an interesting
connection between dynamic problems and distance labeling and also
slightly improves the result of \cite{GavoillePPR04} providing a tight
lower bound for distance labeling in weighted planar graphs (this is discussed
in Section~\ref{sec:techniques}).

Our first result is a conditional \emph{polynomial} lower bound for dynamic shortest paths on planar graphs.
Like several recent results \cite{VassW10,AbboudL13,AbboudGV15,AbboudVY15,Saha15,Dahlgaard16}, our lower bound is based on the APSP conjecture.
Perhaps the best argument for this conjecture is the fact that it has endured decades of extensive algorithmic attacks.
Moreover, due to the known \emph{subcubic equivalences}
\cite{VassW10,AbboudGV15,Saha15}, the conjecture is false \emph{if and only if}
several other fundamental graph and matrix problems can be solved substantially
faster.

\begin{theorem}\label{thm:apsp_tradeoff}
    No algorithm can solve the dynamic APSP problem in planar graphs on $N$
    nodes with amortized query time $q(N)$ and update time $u(N)$ such that
    $q(N)\cdot u(N) = O(N^{1-\eps})$ for any $\eps > 0$ unless Conjecture~\ref{conj:apsp} is
    false. This holds even if we only allow weight updates to $G$.
\end{theorem}

Thus, under the APSP conjecture, there is no hope for a very efficient dynamic
shortest paths algorithm on planar graphs with provable guarantees.
We show that an algorithm achieving $O(n^{1/2-\eps})$ time for both updates and queries is unlikely, implying that the current upper bounds achieving $\tilde{O}(n^{2/3})$ time are not too far from being conditionally optimal. Furthermore,
our result implies that any algorithm with subpolynomial query time must have
\emph{linear} update time (and the other way around). Thus, the naïve algorithm
of simply computing the entire shortest path every time a query is made is
(conditionally) optimal if we want $n^{o(1)}$ update time.

An important property of Theorem \ref{thm:apsp_tradeoff} is that our reduction
does not even violate planarity with respect to a \emph{fixed embedding}.
Thus, we give lower bounds even for \emph{plane graph} problems, which in many cases allow for improved upper bounds over flexible planar graphs (e.g. for reachability\cite{DiksS07,AvrahamKS15}).
Moreover, our graphs are \emph{grid graphs} which are subgraphs of the infinite
grid, a special and highly structured subclass of planar graphs.
Finally, as stated in Theorem~\ref{thm:apsp_tradeoff} our lower bound holds even for the edge weight update model of Abraham \emph{et al.} \cite{AbrahamCDGW16}, where each edge only ever changes its weight to within a factor of $M>1$.
While they obtain fast $\polylog(n)$ time $(1+\eps)$-approximation algorithm in
this model, we show that an exact answer with the same query time likely
requires \emph{linear} update time and that an algorithm with $O(n^{1/2-\eps})$
runtime for both is highly unlikely. Thus, further theoretical restrictions need to be added in order to explain the impressive performance on real road networks.

We also extend Theorem~\ref{thm:apsp_tradeoff} to the case in which we only need
to maintain one $s,t$ distance (the $s,t$-shortest path problem). While this
problem is equivalent to the APSP version in general (as we may connect $s$ and $t$ to any two
nodes $u,v$ we wish to know the distance between) this may violate planarity
and especially a fixed embedding. We show that this problem exhibits similar
trade-offs under Conjecture~\ref{conj:apsp} even if we are only allowed to update
weights. Finally, we note that in the case of directed planar graphs
allowing negative edge weights our techniques can be extended to show the same
hardness result for any approximation under Conjecture~\ref{conj:apsp}.

\medskip

Next, we seek a lower bound for the unweighted version of the problem, which arguably, is of more fundamental interest.
Typically, a conditional lower bound under the APSP conjecture for a weighted problem can be modified into a lower bound for its unweighted version under the Boolean Matrix Multiplication (BMM) Conjecture \cite{DorHZ00,RodittyZ11,VassW10,AbboudV14,AbboudGV15}.
While for combinatorial algorithms the complexity of BMM is conjectured to be cubic, it is known that using algebraic techniques there is an $O(n^\omega)$ algorithm, where $\omega<2.373$ \cite{Vass12,LeGall14}.
When reducing to dynamic problems, however, lower bounds under BMM are often
under a certain online version of BMM for which, Henzinger \emph{et al.}
\cite{HenzingerKNS15} conjecture that there is no truly subcubic algorithms,
even using algebraic techniques.
This \emph{Online Matrix Vector Multiplication} (OMv) Conjecture is stated
formally in Section~\ref{sec:unweighted}.

The OMv conjecture implies strong lower bounds for many dynamic problems on general graphs \cite{HenzingerKNS15}, via extremely simple reductions \cite{AbboudV14,HenzingerKNS15}.
Our next result is a significantly more involved reduction from OMv to dynamic shortest paths on planar graphs, giving unweighted versions of the theorems above.
The lower bounds are slightly weaker but they still rule out algorithms with
subpolynomial update and query times, even in grid graphs.
We remark that all lower bounds under the APSP conjecture in this paper, such
as Theorem~\ref{thm:apsp_tradeoff}, also hold under OMv.

\begin{theorem}\label{thm:apsp_unit}
    No algorithm can solve the dynamic APSP problem in unit weight planar
    graphs on $N$ nodes with amortized query time $q(N)$ and update time $u(N)$
    such that $\max(q(N)^2\cdot u(N),q(N)\cdot u(N)^2) = O(N^{1-\eps})$ for any
    $\eps > 0$ unless the OMv conjecture of \cite{HenzingerKNS15} is false.
    This holds even if we only allow weight updates.
\end{theorem}

For instance, Theorem~\ref{thm:apsp_unit} shows that no algorithm is likely to
have $O(n^{\frac{1}{3}-\eps})$ amortized time for both queries and updates. It
also shows that if we want to have $n^{o(1)}$ for one we likely need
$n^{\frac{1}{2}-o(1)}$ time for the other.

Combined with previous results, our theorems reveal a mysterious phenomenon: there are two contradicting separations between planar graphs and small treewidth graphs, in terms of the time complexity of dynamic problems related to shortest paths (under popular conjectures).
To illustrate these separations, consider the dynamic $s,t$-shortest path problem and the dynamic approximate diameter problem.
For $s,t$-shortest path, planar graphs are much harder, they require
$n^{1/3-o(1)}$ update or query time by Theorem~\ref{thm:apsp_unit} (under OMv), while on
small (polylog) treewidth graphs there is an algorithm achieving polylog
updates and queries \cite{AbrahamCDGW16}.
On the other hand, for approximate diameter, planar graphs are provably
\emph{easier} under the Strong Exponential Time Hypothesis (SETH).
A naive algorithm that runs the known $\tilde{O}(n)$ time static algorithm for $(1+\eps)$ approximate diameter on planar graphs after each update \cite{WeimannY16}, shatters an $n^{2-o(1)}$ SETH-based lower bound for a $(4/3-\delta)$ approximation for diameter on graphs with treewidth $O(\log n)$ \cite{AbboudV14}\footnote{This lower bound follows from observing that the reduction from CNF-SAT to dynamic diameter \cite{AbboudV14} produces graphs with logarithmic treewidth. For more details on an analogous observation w.r.t. the lower bound for diameter in static graphs, see \cite{AbboudVW16}.}.

\medskip

We demonstrate the potential of our framework to yield further strong lower
bounds for important problems in planar graphs by proving such a result
for another well-studied problem in the graph theory literature, namely Maximum
Weight Matching.

Maintaining a maximum matching in general dynamic graphs is a difficult task: the best known algorithm by Sankowski \cite{Sankowski07} has an $O(n^{1.495})$ amortized update time, and it is better than the simple $O(m)$ algorithm (that looks for an augmenting path after every update) only in dense graphs.
Recent results show barriers for much faster algorithms via conjectures like OMv and $3$-SUM \cite{AbboudV14,KopelowitzPP16,HenzingerKNS15,Dahlgaard16}.
To our knowledge, this $O(m)$ update time is the best known for planar graphs and no lower bound is known.
Meanwhile, there has been tremendous progress on approximation algorithms  \cite{OnakR10,BaswanaGS11,NeimanS15,BernS15, BhattaHI15,BernS16,HeTZ14,KopeKPS14,PelegS16, AnandBGS12}, both on general and planar graphs, as well as for the natural \emph{Maximum Weight Matching} (see the references in  \cite{DuanP14} for the history of this variant).
Planar graphs have proven easier to work with in this context: the state of the art deterministic algorithm for maintaining a $(1+\eps)$-maximum matching in general graphs has $O(\sqrt{m})$ update time \cite{GuptaP13}, while in planar graphs the bound is $O(1)$ \cite{PelegS16}.

We show a strong \emph{polynomial} lower bound for Max Weight Matching on
planar graphs, that holds even for bipartite graphs with a fixed embedding into
the plane and even in grid graphs. The lower bound is similar to
Theorem~\ref{thm:apsp_tradeoff} and shows a trade-off between query and update time.

\begin{theorem}\label{thm:mwm}
    No algorithm can solve the dynamic maximum weight matching problem in
    bipartite planar graphs on $N$ nodes with amortized update time $u(N)$ and
    query time $q(N)$ such that $\max(q(N),u(N)) =
    O(N^{\frac{1}{2}-\eps})$ for any $\eps > 0$ unless Conjecture~\ref{conj:apsp} is
    false. Furthermore, if $q(N) \ge u(N)$ the algorithm cannot have $q(N)\cdot
    u(N) = O(N^{1-\eps})$. This holds even if the planar embedding of $G$
    never changes.
\end{theorem}

Finally, we use our framework to show lower bounds for various other problems,
like dynamic girth and diameter. We also argue that our bounds can be
turned into worst-case bounds for incremental and decremental versions of the
same problems.

\subsection{Techniques and relations to distance labeling}\label{sec:techniques}
To prove the results mentioned above we introduce a new framework for reductions to optimization problems on planar graphs.
As mentioned we combine ideas from previous lower bound proofs for dynamic graph problems
with an approach inspired by the framework of Marx for hardness of
parameterized geometric problem (via the Grid Tiling problem)
\cite{Marx07,Marx13} and a graph construction from the research on labelling
schemes by Gavoille \emph{et al.} \cite{GavoillePPR04}.

Gavoille \emph{et al.} \cite{GavoillePPR04} used a family of grid-like graphs
to prove an (unconditional) lower bound of $\Omega(\sqrt{n})$ on the label size
of distance labeling in weighted planar graphs along with a $O(\sqrt{n}\log n)$ upper
bound. (A full discussion of distance labeling schemes is outside the scope of
this paper. For details on this we refer to
\cite{GavoillePPR04,AlstrupGHP16,AlstrupDKP15}).
In this paper we generalize their family of graphs to a family of
grid graphs capable of representing general matrices with weights in $[\poly(n)]$
via shortest paths distances. Using our construction with the framework of
\cite{GavoillePPR04}, we obtain a \emph{tight} $\Omega(\sqrt{n}\log n)$ lower
bound on the size of distance labeling in weighted planar graphs (and even grid
graphs).

Our main approach works by reducing from the
$(\min,+)$-Matrix-Multiplication problem which is known to be equivalent to APSP (see \cite{VassW10}):
Given two $n \times n$ matrices $A,B$  with entries in $[poly(n)]$, compute a matrix $C = A \oplus B$ such that $C[i,j] = \min_{k\in[n]} A[i,k] + B[k,j]$.
 By concatenating grid
graphs from the family described above we are able to represent
one of the matrices in the product and we can then simulate the multiplication
process via updates and shortest paths queries.

In a certain intuitive sense, our connection between dynamic algorithms and labeling schemes is the reverse direction of the one shown by Abraham \emph{et al.} \cite{AbrahamCG12} to obtain their $\tilde{O}(\sqrt{N})$ update time $(1+\eps)$-approximation algorithm for dynamic APSP.
Their algorithm utilizes a clever upper bound for the so-called \emph{forbidden
set distance labeling} problem,
while our lower bound constructions have a clever lower bound for labeling schemes embedded in them.

\section{A grid construction}\label{sec:grid}
In order to reduce to problems on planar graphs we will need a planar
construction, which is able to capture the complications of problems like OMv
and APSP. To do this we will employ a grid construction based on the one used
in \cite{GavoillePPR04} to prove lower bounds on distance labeling for planar
graphs. Our construction takes a matrix as input and produces a grid graph
representing that matrix. We first present a boolean version similar to the one
from \cite{GavoillePPR04} and then modify it to obtain a version taking
matrices with integer entries as input. This modified matrix also immediately
leads to a tight $\Omega(\sqrt{n}\log n)$ lower bound for distance labeling in
planar graphs with weights in $[poly(n)]$ when combined with the framework of
\cite{GavoillePPR04}.

\begin{definition}\label{defn:grid}
    Let $M$ be a boolean $R\times C$ matrix. We will call the following
    construction the \emph{grid embedding} of $M$:

    Let $G_M$ be a rectangular grid graph with $R$ rows and $C$ columns. Denote
    the node at intersection $(i,j)$ by $u_{i,j}$ ($u_{1,1}$ is top-left and
    $u_{R,C}$ is bottom-right). Add $C$ nodes $a_1,\ldots, a_C$ and edges
    $(u_{1,j}, a_j)$ above $G_M$. Similarly add the nodes $b_1,\ldots, b_R$ and
    edges $(u_{i,C},b_i)$ to the right of $G_M$. Now subdivide each vertical
    edge adding the node $v_{i,j}$ above $u_{i,j}$, and subdivide each
    horizontal edge adding the node $w_{i,j}$ to the right of $u_{i,j}$.
    Finally, for each entry of $M$ such that $M_{i,j} = 1$ add the node
    $x_{i,j}$ and edges $(v_{i,j},x_{i,j})$ and $(w_{i,j},x_{i,j})$ to the
    graph.

    The weights of $G_M$ are as follows: Each edge $(u_{i,j}, v_{i+1,j})$ and
    $(a_j, v_{1,j})$ has weight $2j-1$. Each edge $(w_{i,j}, u_{i+1,j})$ and
    $(w_{i,C}, b_i)$ has weight $2R-2$. The edge $(u_{i,j},w_{i,j})$ has weight
    $2$. All remaining edges have weight $1$.
\end{definition}

We will call the two-edge path $v_{i,j}\to x_{i,j}\to w_{i,j}$ a
\emph{shortcut} from $v_{i,j}$ to $w_{i,j}$ as it has length $1$ less than the
path $v_{i,j}\to u_{i,j}\to w_{i,j}$. Clearly, the grid embedding of a $R\times
C$ matrix has $O(RC)$ nodes. It is also easy to see that such a grid embedding
is a subgraph of a $2R+1\times 2C+1$ rectangular grid. The construction of
Definition~\ref{defn:grid} for a $3\times 3$ matrix can be seen in Figure~\ref{fig:grid}.

\begin{figure}[htbp]
    \centering
    \includegraphics[width=.65\textwidth]{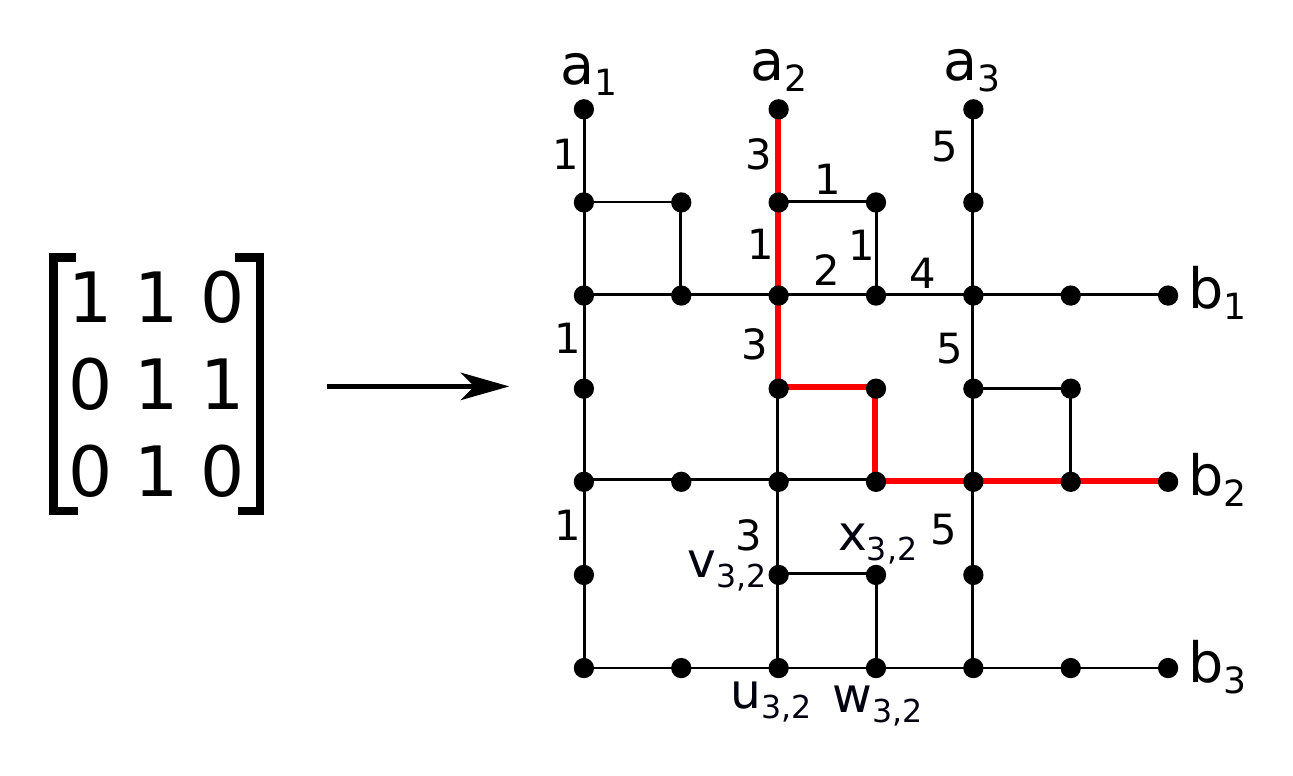}
    \caption{Illustration of the construction of Definition~\ref{defn:grid}. The shortest
    path from $a_2$ to $b_2$ is highlighted in red. Most edge weights are
    omitted for clarity.}
    \label{fig:grid}
\end{figure}

\begin{proposition}\label{prop:grid_dist}
    Let $M$ be a boolean $R\times C$ matrix and let $G_M$ be its grid embedding
    as defined in Definition~\ref{defn:grid}. Then for any $1\le i\le R,1\le j\le C$ and
    $i< k\le R$ the shortest path distance from $u_{i,j}$ to $b_k$ is exactly
    \[
        (k-i)\cdot2j + 2R\cdot (C-j+1)
    \]
    if $M_{k,j} = 0$ and
    \[
        (k-i)\cdot2j + 2R\cdot (C-j+1) - 1
    \]
    otherwise.
\end{proposition}
\begin{proof}
    Consider any shortest path from any $u_{i,j}$ to $b_k$. Such a path must
    always go either ``right'' or ``down'' (if $i=k$ the path must always go
    right). Essentially for every step to the left we pay at least $4R-1$
    but can at most save $2R$: paying $2R$ going left and $2R$ going right,
    possibly saving $1$ with a shortcut, and saving $2$ for each vertical
    edge.

    Now we will show the claim by induction on the sum $i+j$. Clearly, for
    $u_{R,C}$ to $b_R$ the distance is exactly $2R$. Now consider $u_{i,j}$ and
    assume $k > i$ as the case of $k=i$ is trivial. There are three cases to
    consider:
    \begin{enumerate}
        \itemsep-2pt
        \item The path from $u_{i,j}$ goes through $w_{i,j}$ and then
            $u_{i,j+1}$. By the induction hypothesis this path has length at
            least
            \begin{align*}
                2R + (k-i)\cdot 2(j+1) + 2R\cdot (C - j) - 1
                \ge 2R \cdot (C - j + 1) + (k-i)\cdot 2j\ .
            \end{align*}
        \item The path from $u_{i,j}$ goes through $v_{i+1,j}$ and then
            $x_{i+1,j}$, $w_{i+1,j}$, and $u_{i+1,j+1}$. This path is only
            available if $M_{i+1,j} = 1$. If $k = i+1$, this distance is exactly
            \[
                2j-1 + 2R \cdot(C - j + 1)\ .
            \]
            Otherwise, by the induction hypothesis, it is at least
            \[
                (k-i-1)\cdot 2(j+1) + 2R\cdot(C - j) - 1 + 2j +
                2R - 1 \ge (k-i)\cdot 2j +2R\cdot(C - j + 1)
            \]
            for $k > i+1$
        \item The path from $u_{i,j}$ goes through $v_{i+1,j}$ and then
            $u_{i+1,j}$. By the induction hypothesis, if $M_{k,j} = 1$, the
            length of this path is
            \[
                (k-i-1)\cdot 2j + 2R\cdot (C - j + 1) - 1 + 2j =
                (k-i)\cdot 2j + 2R\cdot (C-j+1) - 1
            \]
            and otherwise it is
            \[
                (k-i-1)\cdot 2j + 2R\cdot(C-j+1) + 2j\ .
            \]
    \end{enumerate}
    It is easy to verify that taking the path down and right as illustrated in
    Figure~\ref{fig:grid} gives exactly the distances in the proposition, finishing the
    proof.
\end{proof}

The following useful property of our grid construction follows.

\begin{corollary}\label{cor:grid_dist}
    Let $M$ and $G_M$ be as in Proposition~\ref{prop:grid_dist}. Then for any $1\le k\le
    R,1\le j\le C$, the distance between $a_j$ and $b_k$ in $G_M$ is exactly
    determined by whether $M_{k,j} = 1$. In this case the distance is
    $2R\cdot(C-j+1) + 2jk - 1$ and it is $2R\cdot(C-j+1) + 2jk$ otherwise.
\end{corollary}

The following generalization for matrices with integer weights will be useful
when reducing from APSP.

\begin{definition}\label{defn:grid_weight}
    Let $M$ be a $R\times C$ matrix with integer weights in $\{0,\ldots,
    X\}$. We will call the following construction the \emph{grid embedding} of
    $M$.

    Let $G_M$ be the grid embedding from Definition~\ref{defn:grid} for the all ones
    matrix of size $R\times C$ and multiply the weight of each edge by
    $X^2$. Furthermore, for each edge $(v_{i,j},x_{i,j})$ increase its weight
    by $M_{i,j}$.
\end{definition}

\begin{corollary}\label{cor:grid_dist_wt}
    Let $M$ be a $R\times C$ matrix with integer weights in $\{0,\ldots,X\}$
    and let $G_M$ be its grid embedding. Then for any $1\le k\le R,1\le j\le C$,
    the distance between $a_j$ and $b_k$ in $G_M$ is exactly
    \[
        X^2\cdot(2R\cdot(C-j+1) + 2jk -1) + M_{k,j}
    \]
\end{corollary}

Corollary~\ref{cor:grid_dist_wt} follows from Corollary~\ref{cor:grid_dist} by observing that any
path from $a_j$ to $b_k$ not using the shortcut at intersection $(k,j)$ has
distance at least $X^2\cdot(2R\cdot(C-j+1)+2jk)$ and since $M_{k,j} < X^2$
this distance is longer than using the shortcut. We remark that it would have been sufficient to multiply the weights by $(X+1)$ instead of $X^2$, but we do so to simplify a later argument.

\section{Hardness of dynamic APSP in planar graphs}\label{sec:apsp}
We will first show the following, simpler theorem and then generalize it to
show trade-offs between query and update time.

\begin{theorem}\label{thm:apsp_simple}
    No algorithm can solve the dynamic APSP problem in planar graphs on $N$
    nodes with amortized update and query time $O(N^{\frac{1}{2}-\eps})$ for
    any $\eps > 0$ unless Conjecture~\ref{conj:apsp} is false. This holds even
    if only weight updates are allowed.
\end{theorem}

The main idea in proving Theorem~\ref{thm:apsp_simple} is to reduce from the APSP problem by
first reducing to $(\min,+)$-Matrix-Mult and use the grid construction
from Section~\ref{sec:grid} to represent the matrices to be multiplied. We then
perform several shortest paths queries to simulate the multiplication process.
Below, we first present a naïve and faulty approach explaining the main ideas
of the reduction. We then show how to mend this approach giving the desired
result.

\begin{attempt}\label{att:apsp_bad}
    Consider the following algorithm for solving an instance, $A\oplus B$ of
    the $(\min,+)$-Matrix-Mult problem, where $A$ and $B$ are $n\times n$
    matrices. We may assume that $A$ and $B$ have integer weights in
    $\{0,\ldots, X\}$ for some $X = poly(n)$.

    We let the initial graph of the problem be the grid embedding $G_B$ of $B$
    according to Definition~\ref{defn:grid_weight} along with a special vertex $t$. Also
    add the edges $(b_k,t)$ for each $1\le k\le n$.
    Now we wish to construct $C = A\oplus B$ one row at a time. Such a row is a
    $(\min,+)$-product of a row in $A$ and the entire matrix $B$. Thus, for
    each row, $i$, of $A$ we have a phase as follows:
    \begin{enumerate}
        \itemsep-2pt
        \item For each $1\le k\le n$ update the weight of the edge $(b_k, t)$
            to be $A_{i,k}$.
        \item For each $1\le j\le n$ query the distance between $a_j$ and $t$.
    \end{enumerate}

    The idea of each phase is that the distance between $a_j$ and $t$ should
    correspond to the value of $C_{i,j} = \min_k A_{i,k}+B_{k,j}$. Observe, that the distance from $a_j$
    to $t$ using the edge $(b_k,t)$ is exactly
    \[
        X^2\cdot(2n\cdot(n-j+1) + 2jk -1) + B_{k,j} + A_{i,k}
            \]
        by Corollary~\ref{cor:grid_dist_wt}. The dominant term in this expression increases with $k$ and thus no
    matter what $B_{k,j}$ and $A_{i,k}$ are (for $k>1$), the shortest path from $a_j$ to $t$
    will simply pick $k=1$ minimizing the above expression.
     If we instead set
    the weight of each edge $(b_k,t)$ to $X^2\cdot 2j(n-k) + A_{i,k}$ we get
    the distance of using this edge to be
    \[
        X^2\cdot(2n\cdot(n-j+1) + 2jk -1) + B_{k,j} + A_{i,k} + X^2\cdot
        2j(n-k) = X^2\cdot(2n(n+1) -1) + B_{k,j} + A_{i,k}\ .
    \]
    It follows that the shortest path from $a_j$ to $t$ is free to pick any $k$ while only affecting the $B_{k,j} + A_{i,k}$ term, which means that the shortest distance will be achieved by picking the $k$ minimizing this term, which would give us exactly $X^2\cdot(2n(n+1)-1) + C_{i,j}$. This
    approach therefore allows us to correctly calculate $C = A\oplus B$.
    However, the weight of the edge $(b_k,t)$ now depends on which $a_j$
    we are querying implying that we have to update this weight for each $a_j$
    leading to a total of $O(n^3)$ updates. By using this approach we are thus
    not able to make any statement about the time required for updates. We may
    try to assign edges and weights differently, but such approaches run into
    similar issues.

    Observe that the graph created has $N = O(n^2)$ nodes. Thus, if we were
    able to perform only $O(n^2)$ total queries and updates the result of
    Theorem~\ref{thm:apsp_simple} would follow.\hfill $\Diamond$
\end{attempt}

In order to circumvent this dependence on $j$ when assigning weights to the
edges $(b_k,t)$ we instead replace $t$ by another grid whose purpose is to
``normalize'' the distance for each $a_j$.
 By doing this we can connect the
grids with edges whose weight is independent of $j$.
This step deviates significantly from the construction of \cite{GavoillePPR04} and is inspired by the grid tiling framework of Marx \cite{Marx07,Marx13}.

\begin{proof}[Proof of Theorem~\ref{thm:apsp_simple}]
    We follow the same approach as in Attempt~\ref{att:apsp_bad}, but with a few
    changes. Define the initial graph $G$ as follows: Let $G_B$ be as before
    and let $G'_B$ be the grid embedding of $B$ mirrored along the vertical
    axis with all shortcuts removed. Now for each $1\le k\le n$ add the edge
    $(b_k,b'_k)$ and define $G$ to be this graph.

    Now we perform a phase for each row $i$ of $A$ as follows:
    \begin{enumerate}
        \itemsep-2pt
        \item For each $1\le k\le n$ set the weight of the edge $(b_k,b'_k)$ to
            be $X^2\cdot (2(n+1)(n-k)) + A_{i,k}$.
        \item For each $j$ query the distance between $a_j$ and $a'_{n-j+1}$.
    \end{enumerate}
    An example of this construction for $n=3$ can be seen in Figure~\ref{fig:apsp_red}.
    \begin{figure}[htbp]
        \centering
        \includegraphics[width=.8\textwidth]{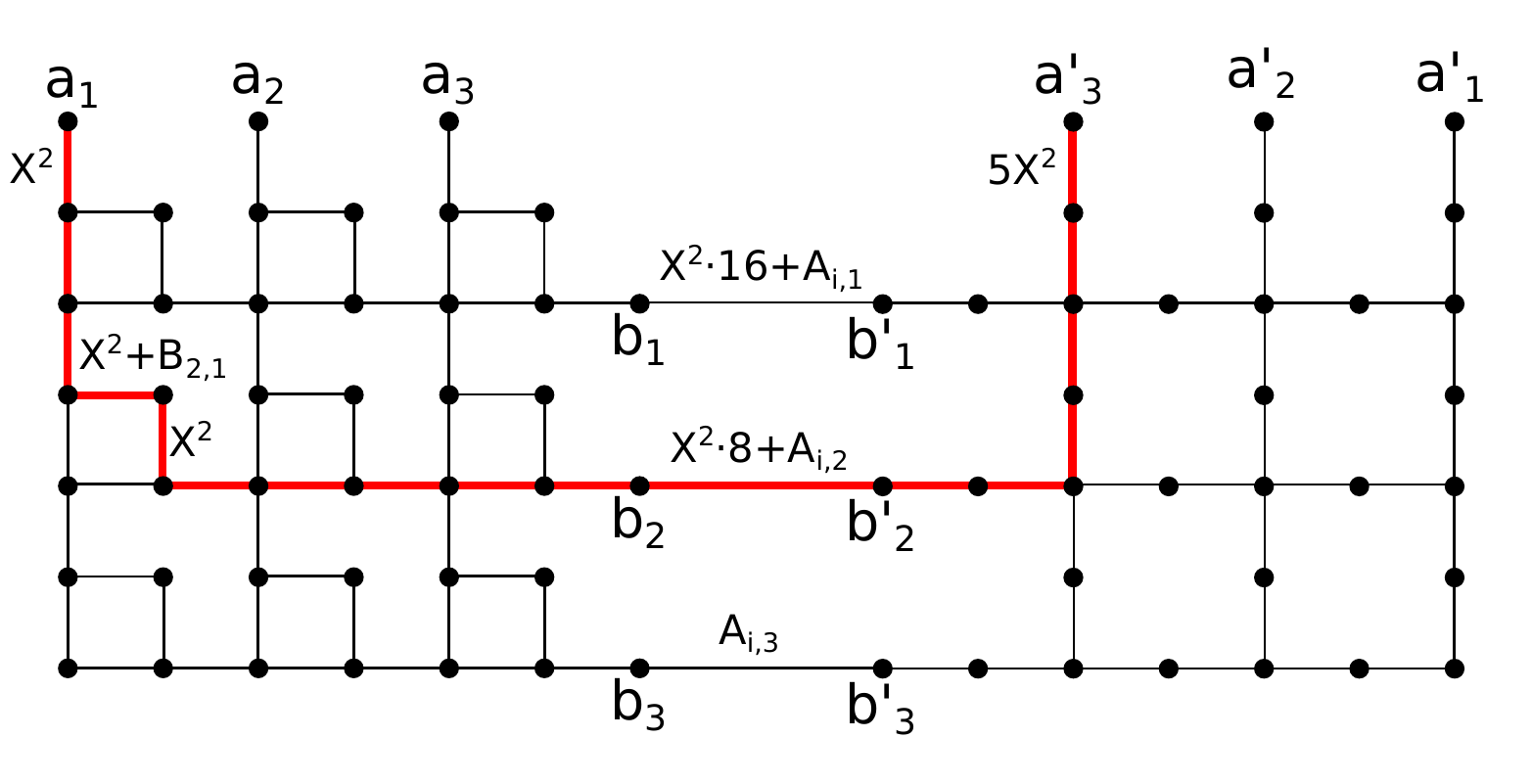}
        \caption{Example of a phase in the graph $G$ in the reduction of
        Theorem~\ref{thm:apsp_simple}. The highlighted path illustrates a shortest
        path between $a_1$ and $a'_3$ as an example of a query. Most edge
        weights have been omitted for clarity. }
        \label{fig:apsp_red}
    \end{figure}

    From the query between nodes $a_j$ and $a'_{n-j+1}$ above during phase $i$
    we can determine the entry $C_{i,j}$ of the output matrix. To see this,
    consider the distance from $a_j$ to $a'_{n-j+1}$ at the time of query. This
    path has to go via some edge $(b_k,b'_k)$. From Corollary~\ref{cor:grid_dist_wt} we
    know that this distance is exactly
    \begin{align*}
        d_G(a_j,a'_{n-j+1})
        &= d(a_j,b_k) + w(b_k,b'_k) + d(b'_k,a'_{n-j+1}) \\
        &= X^2\cdot (2n\cdot(n-j+1) + 2jk -1) + B_{k,j} + X^2\cdot(2(n+1)(n-k))
        + A_{i,k} \\
        &\qquad + X^2\cdot(2n\cdot j + 2(n-j+1)k) \\
        &= X^2\cdot 2n(n+1) + X^2\cdot2k(n+1) + X^2\cdot 2(n+1)(n-k) + B_{k,j}
        + A_{i,k} - X^2 \\
        &= X^2\cdot 4n(n+1) - X^2 + B_{k,j} + A_{i,k}\ .
    \end{align*}
    The crucial property that our construction achieves is that the dominant term of this expression is independent of $k$.
    Thus, the shortest path will choose to go through the edge $(b_k,b'_k)$ that minimizes $B_{k,j} + A_{i,k}$, implicitly giving us $C_{i,j}$.
	Subtracting $X^2\cdot(4n(n+1) - 1)$ from the queried
    distance gives exactly the value of $C_{i,j}$ and the algorithm therefore
    correctly computes $C$.

    Following the analysis from Attempt~\ref{att:apsp_bad} we have that any algorithm
    with an amortized running time of $O(N^{\frac{1}{2}-\eps})$ for
    both updates and queries contradicts Conjecture~\ref{conj:apsp}.
\end{proof}

\subsection{Trade-offs}\label{sec:tradeoffs}
Theorem~\ref{thm:apsp_simple} above shows that no algorithm can perform both updates
and queries in amortized time $O(N^{\frac{1}{2}-\eps})$ unless Conjecture~\ref{conj:apsp}
is false. We will now show how to generalize these ideas to show
Theorem~\ref{thm:apsp_tradeoff}.
\begin{proof}[Proof of Theorem~\ref{thm:apsp_tradeoff}]
    The proof follows the same structure as the proof for
    Theorem~\ref{thm:apsp_simple}, but instead of reducing from $(\min,+)$-Matrix-Mult
    on $n\times n$ matrices we reduce from an unbalanced version.

    Let $A$ and $B$ be $n\times n^\beta$ and $n^\beta\times n^\alpha$ matrices
    respectively for some $0 < \alpha,\beta \le 1$. We define the initial graph
    $G$ from $B$ in the same manner as in Theorem~\ref{thm:apsp_simple}. We then have
    a phase for each row $i$ of $A$ as follows:
    \begin{enumerate}
        \itemsep-2pt
        \item For each $1\le k\le n^\beta$ set the weight of the edge
            $(b_k,b'_k)$ to be $X^2\cdot (2(n^\alpha+1)(n^\beta-k)) + A_{i,k}$.
        \item For each $1\le j\le n^\alpha$ query the distance between $a_j$
            and $a'_{n^\alpha-j+1}$.
    \end{enumerate}
    The entry $C_{i,j}$ is exactly the distance $d_G(a_j,a'_{n^\alpha-j+1})$
    from the $i$th phase minus $X^2\cdot (4n^\beta(n^\alpha+1) - 1)$.
    The correctness of the above reduction follows directly from the proof of
    Theorem~\ref{thm:apsp_simple} as well as Corollary~\ref{cor:grid_dist_wt}.

    Now observe that the graph $G$ from the above reduction has $N =
    \Theta(n^{\alpha+\beta})$ nodes and we perform a total of $O(n^{1+\alpha})$
    queries and $O(n^{1+\beta})$ updates\footnote{We also perform
    $O(n^{\alpha+\beta})$ updates to create the initial graph (depending on the
    model), however we will choose $\alpha$ and $\beta$ such that this term is
    dominated.} -- that is, at most $O(n)$ updates per row and $O(n)$ queries per
    column. Any algorithm solving this problem must use total time
    $n^{1+\alpha+\beta-o(1)}$ time unless Conjecture~\ref{conj:apsp} is false. It follows
    that either updates must take $n^{\alpha-o(1)}$ amortized time or
    queries must take $n^{\beta-o(1)}$ amortized time.

    Assume now that an algorithm exists such that queries take $O(N^\gamma)$
    amortized time for any $0 < \gamma < 1$. We wish to show that this algorithm
    cannot perform updates in amortized time $O(N^{1-\gamma-\eps})$ for any
    $\eps > 0$. Pick $\beta = \gamma + \eps/2$ and set $\alpha = 1 - \beta$. We
    now use the above reduction to create a dynamic graph $G$ with $N =
    O(n^{\alpha+\beta}) = O(n)$ nodes. Since queries do not take $n^{\beta -
    o(1)}$ time it follows from the above discussion that updates must take
    $n^{\alpha - o(1)} = n^{1-\gamma -\eps/2 - o(1)}$ time. Since this is
    polynomially greater than $O(N^{1-\gamma-\eps})$ the claim follows.
\end{proof}

\section{Hardness of dynamic maximum weight matching in bipartite planar
graphs}
In this section we will demonstrate the generality of our reduction framework
by showing Theorem~\ref{thm:mwm}.
\begin{proof}[Proof of Theorem~\ref{thm:mwm}]
    We start by showing how to reduce from $(\min,+)$-Matrix-Mult to minimum
    weight perfect matching, where the weight of such a matching corresponds to
    the shortest path distance between $a_j$ and $a'_{n-j+1}$ similar to the
    proof of Theorem~\ref{thm:apsp_tradeoff}. We then describe how to use this
    reduction further to get a problem instance for maximum weight matching.

    Let $A,B$ be an instance to the $(\min,+)$-Matrix-Mult problem of sizes
    $n\times n^\beta$ and $n^\beta\times n^\alpha$ respectively. Consider the
    grid embedding $G_B$ of $B$. We first replace each node of $G_B$ by two
    nodes connected by an edge of weight $0$.
    For $a_j$, $u_{i,j}$, $x_{i,j}$, and $v_{i,j}$ denote the corresponding
    nodes with superscript $d$ and $u$ (for ``down'' and ``up''). For $b_i$ and
    $w_{i,j}$ denote the corresponding nodes with superscript $l$ and $r$ (for
    ``left'' and ``right''). Now, for each original edge in $G_B$ we replace it
    as follows keeping its weight:
    \begin{multicols}{2}
    \begin{itemize}
        \itemsep-2pt
        \item $(u_{i,j},v_{i,j})\rightarrow (u_{i,j}^u, v_{i,j}^d)$
        \item $(u_{i,j},w_{i,j})\rightarrow (u_{i,j}^d, w_{i,j}^l)$
        \item $(u_{i,j},v_{i+1,j})\rightarrow (u_{i,j}^d, v_{i+1,j}^u)$
        \item $(u_{i,j},w_{i,j-1})\rightarrow (u_{i,j}^u, w_{i,j-1}^r)$
        \item $(v_{i,j},x_{i,j})\rightarrow (v_{i,j}^d, x_{i,j}^u)$
        \item $(x_{i,j},w_{i,j})\rightarrow (x_{i,j}^d, w_{i,j}^l)$
        \item $(a_j,v_{0,j})\rightarrow (a_j^d,v_{0,j}^u)$
        \item $(w_{i,C},b_i)\rightarrow (w_{i,C}^r,b_i^l)$
    \end{itemize}
    \end{multicols}
    This construction is illustrated in Figure~\ref{fig:match_grid}. We call this
    modified grid structure $\bar{G}_B$. Observe that there are no edges
    between ``up'' and ``left'' vertices or between ``down'' and ``right''. It
    follows that the graph is bipartite and that these two sets of nodes make
    up the two partitions.
    \begin{figure}[htbp]
        \centering
        \includegraphics[width=.8\textwidth]{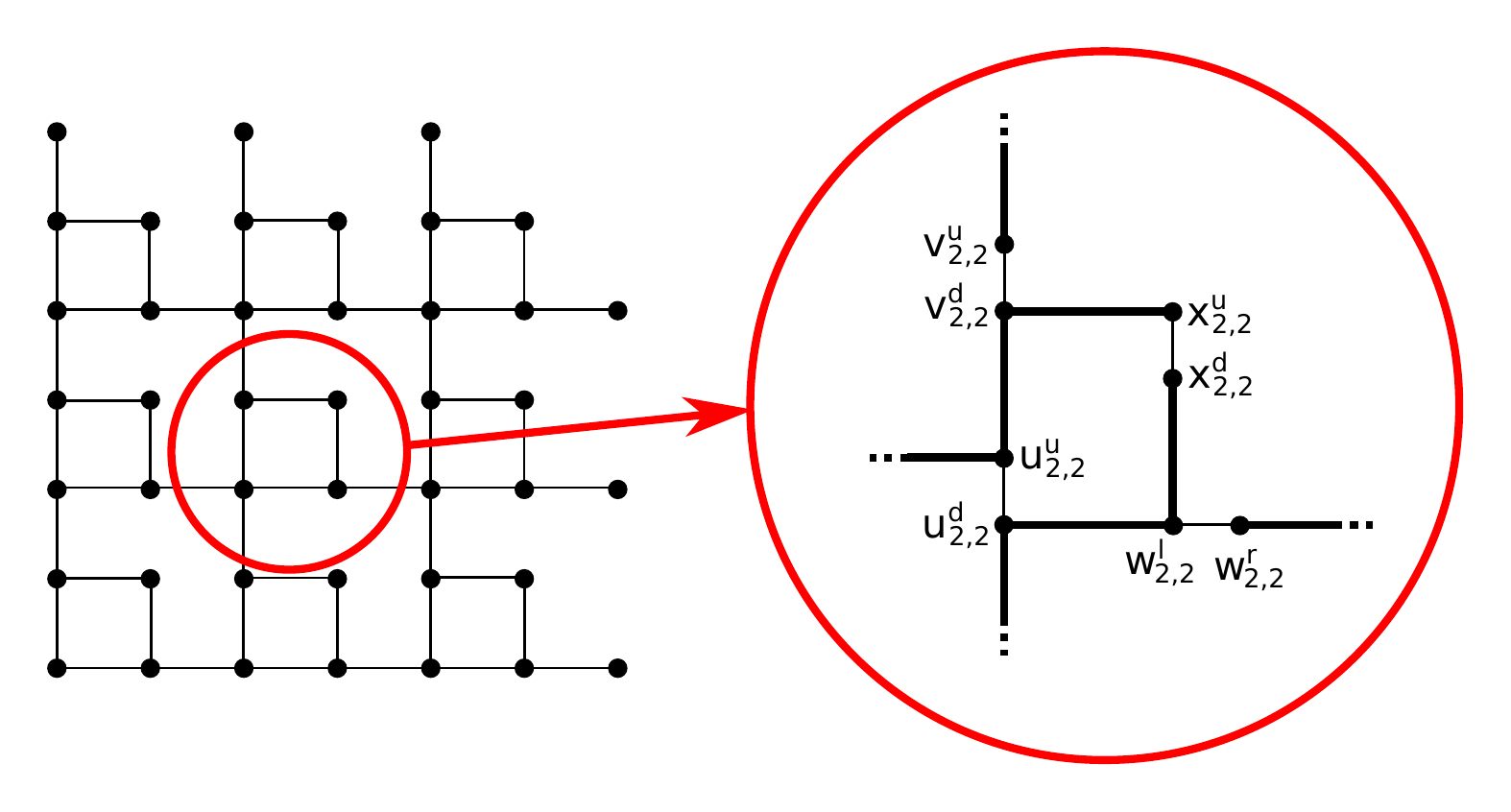}
        \caption{Grid construction for minimum weight perfect matching. Thick
        edges correspond to original edges and have the same weight as in
        $G_B$. Thin edges have weight $0$.}
        \label{fig:match_grid}
    \end{figure}

    We now replace the grids $G_B$ and $G'_B$ by $\bar{G}_B$ and $\bar{G}'_B$
    in the initial graph $G$ from the proof of Theorem~\ref{thm:apsp_tradeoff}. The
    edges $(b_k,b'_k)$ are replaced by $(b^r_k,b^{\prime l}_k)$.
    We will use the following observation.
    \begin{proposition}\label{prop:matching}
        The graph resulting from joining two grids $\bar{G}_B$ in the way of
        Figure~\ref{fig:apsp_red} has a unique perfect matching.
    \end{proposition}

    \begin{proof}
    It is easy to see that simply matching all weight $0$ edges gives a
    perfect matching. Thus we need to show that this is the only perfect
    matching. We will show the claim by a simple ``peeling'' argument.

    Observe that $a_1^u$ only has one incident edge, so the edge
    $(a_1^u,a_1^d)$ must be in any perfect matching and we may ``peel'' away
    these two nodes. It now follows that $v_{1,1}^u$ only has one adjacent
    edge, so $(v_{1,1}^u,v_{1,1}^d)$ has to be in any perfect matching and we
    may peel away these nodes. Now $u_{1,1}^u$ only has one adjacent edge and
    so on for $v_{2,1}^u$, $u_{2,1}^u$, etc. This peels away the entire
    first column. Now each $x^u_{i,1}$ has only one adjacent edge matching
    this leaves each $w_{i,1}^l$ with only one adjacent edge. Peeling these
    nodes away leaves us with a smaller grid and we may start the argument
    over with $a_2$.

    By doing this we see that the edge joining $b_i^r$ and $b_i^{\prime l}$
    cannot be in a perfect matching as $(b_i^r,b_i^l)$ has to be. Thus
    we can repeat the same argument on the second grid.
\end{proof}

    We now add two additional nodes $s$ and $t$ to the initial graph and perform
    a phase for each row $i$ of $A$ as follows:
    \begin{enumerate}
        \itemsep-2pt
        \item For each $1\le k\le n^\beta$ set the weight of the edge
            $(b^r_k,b^{\prime l}_k)$ to be $X^2\cdot (2(n^\alpha+1)(n^\beta-k))
            + A_{i,k}$.
        \item For each $1\le j\le n^\alpha$ do the following three steps: 1)
            add the edges $(s,a^u_j)$ and $(t,a^{\prime u}_{n^\alpha-j+1})$, 2)
            query the minimum weight perfect matching, 3) delete the two edges.
    \end{enumerate}
    Since the edges $(s,a^u_j)$ and $(t,a^{\prime u}_{n^\alpha-j+1})$ have to
    be in any perfect matching this leaves $a^d_j$ and $a^{\prime
    d}_{n^\alpha-j+1}$ unmatched. Any perfect matching now has to ``connect''
    these two nodes by a path of original (weight $>0$) edges.
    The weight of a perfect matching in $G$ then corresponds to the length of a
    shortest path from $a_j$ to $a'_{n^\alpha-j+1}$ in the graph from the proof
    of Theorem~\ref{thm:apsp_tradeoff}. It follows that we get the same trade-offs for
    minimum weight perfect matching as for APSP with the exception that the
    trade-off only holds when $q(N)\ge u(N)$ since we perform $O(1)$ updates
    for each query.

    To show the same result for maximum weight matching we may simply perform
    the following two changes: 1) pick a sufficiently large integer $y$ and set
    the weight of each edge to $y$ minus its weight in the above reduction, and
    2) when adding the edges $(s,a^u_j)$ and $(t,a^{\prime u}_{n^\alpha-j+1})$
    assign them weight $y^2$ such that any maximum weight matching has to
    include these two edges and will have weight
    \[
        y^2 + \frac{N-4}{2}\cdot y - d_{G*}(a_j,a'_{n^\alpha-j+1})\ ,
    \]
    where $G*$ denotes the corresponding graph in the proof of
    Theorem~\ref{thm:apsp_tradeoff}.
\end{proof}

\section{Unweighted}\label{sec:unweighted}
The proofs of the previous sections rely heavily on the weighted grid from
Section~\ref{sec:grid}. We may generalize the ideas to the unweighted case by instead
using the grid of Definition~\ref{defn:grid} and subdividing the edges giving us somewhat
weaker bounds. This gives us Theorem~\ref{thm:apsp_unit}.

The problem we reduce from is the online matrix-vector problem from
\cite{HenzingerKNS15}. We may define this problem as follows: Let $M$ be a
$n\times n$ matrix and let $v^1,\ldots,v^n$ be $n$ boolean vectors arriving in
an online fashion. The task is to pre-process $M$ such that we can output the
product $Mv^i$ for each $i$ before seeing $v^{i+1}$. It was conjectured in
\cite{HenzingerKNS15} that this problem takes $n^{3-o(1)}$ time, while the best known upper bound is $n^3/2^{\Omega(\sqrt{\log{n}})}$ \cite{LarsenW16}.
Known reductions from \cite{VassW10}
show that this conjecture implies a $n^{1+\alpha+\beta-o(1)}$ bound for the
following problem: Let $\alpha,\beta > 0$ be fixed constants and let $M$ be a
boolean $n^\beta\times n^\alpha$ matrix (see \cite{HenzingerKNS15} for the details). After preprocessing $M$, $n$ boolean
vector pairs $(u^1,v^1), \ldots, (u^n,v^n) \in \{0,1\}^{n^\beta} \times \{0,1\}^{n^\alpha} $ arrive one at a time and the task
is to compute $(u^i)^T M v^i$ before being presented with the $i+1$th vector
pair for every $i$. We will use this problem called \emph{the OuMv problem} to
reduce to unit weight dynamic APSP in planar graphs below.

\begin{proof}[Proof of Theorem~\ref{thm:apsp_unit}]
    Consider the reduction from Theorem~\ref{thm:apsp_tradeoff} using a $n^\beta\times
    n^\alpha$ grid. We will use a similar approach to solve the OuMv problem
    below.

    Let $M$ be the $n^\beta\times n^\alpha$ matrix of the OuMv problem and
    create $G_M$ according to Definition~\ref{defn:grid} (note that this grid embedding
    is different from the one used in the proof of Theorem~\ref{thm:apsp_tradeoff}). We also add
    $G'_M$ similarly to the proof of Theorem~\ref{thm:apsp_tradeoff}. We then
    subdivide each edge into a path of the same length. We also add to $G$
    a path of length $2(n^\alpha+1)(n^\beta-k)$ connecting $b_k$ and $b'_k$ for
    each $1\le k\le n^\beta$. We then disconnect $b_k$ and $b'_k$ from this
    path.

    We perform a phase as follows for each vector pair $(u^i,v^i)$:
    \begin{enumerate}
        \item For each $k$ such that $u^i_k = 1$ connect $b_k$ and $b'_k$ to
            their respective path.
        \item For each $j$ such that $v^i_j = 1$ query the distance from $a_j$
            to $a'_{n^\alpha-j+1}$.
        \item Remove all the edges added in step 1.
    \end{enumerate}
    If the answer to any of the queries during the $i$th phase is $4n(n+1)-1$
    the answer to the $i$th product is $1$ and otherwise the answer is $0$.
    This follows from Corollary~\ref{cor:grid_dist} in the same way as
    Theorem~\ref{thm:apsp_tradeoff}.

    By subdividing the edges we get a graph with $N = O(n^{2\beta+\alpha} +
    n^{2\alpha+\beta})$ nodes. We perform $O(n^{1+\alpha})$ queries and
    $O(n^{1+\beta})$ updates. It follows from the OMv conjecture that the
    entire process must take $n^{1+\beta+\alpha-o(1)}$ time, thus either
    updates take $n^{\alpha-o(1)}$ time or queries take $n^{\beta-o(1)}$ time.

    We will assume that $q(N) \ge u(N)$ and note that the other case follows
    symmetrically. Assume that some algorithm can perform queries in $N^\gamma$
    for some $\frac{1}{3} \le \gamma < \frac{1}{2}$. We wish to show that this
    algorithm cannot perform updates in time $N^{1-2\gamma-\eps}$ for any
    $\eps>0$. To do this, pick $\beta = \gamma + \eps/3$ and $\alpha = 1 -
    2\beta$. Note that $\beta \ge \alpha$ (corresponding to $q(N)\ge u(N)$).
    Thus the graph has $N = O(n^{2\beta+\alpha}) = O(n)$ nodes. It now follows
    by the above discussion that the algorithm cannot perform updates faster
    than $n^{\alpha-o(1)} = N^{1-2\gamma - 2\eps/3 + o(1)}$ which proves the
    claim.

    Finally, observe that by using an $n\times n$ matrix in the above reduction
    (i.e. $\alpha=\beta=1$) we see that at least one of updates and queries
    have to take $n^{\frac{1}{3}-o(1)}$ amortized time (similar to
    Theorem~\ref{thm:apsp_simple}).
\end{proof}

To see that we may do the above reduction while keeping the dynamic graph
$G$ as a grid graph, observe that we may multiply the weight of each edge
before subdividing by a large enough constant and then ``zig-zag'' the
subdivided edges in order to fit the grid structure. This is illustrated in
Figure~\ref{fig:grid_unw}.
\begin{figure}[htbp]
    \centering
    \includegraphics[width=.6\textwidth]{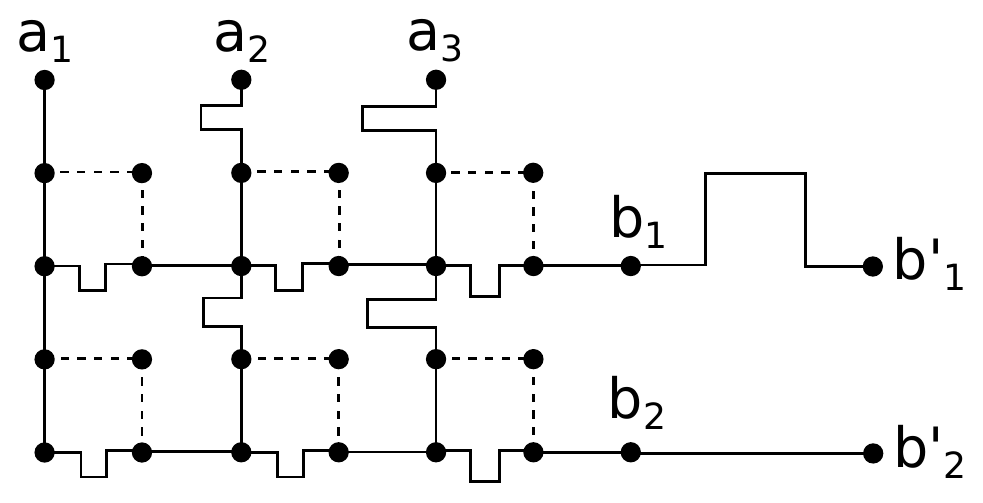}
    \caption{Illustration of the grid created in the proof of
    Theorem~\ref{thm:apsp_unit}. Dashed edges correspond to possible shortcuts. Note
    that the lengths of the edges are not to scale!}
    \label{fig:grid_unw}
\end{figure}

\section{Dynamic $s,t$-shortest path and related problems}
In Section~\ref{sec:apsp} we showed a lower bound for the trade-off between query and
update time for dynamic APSP in grid graphs conditioned on Conjecture~\ref{conj:apsp}.
Here we will argue that the proof of Theorem~\ref{thm:apsp_tradeoff} can be extended
to show similar lower bounds for dynamic problems, where the algorithm only
needs to maintain a single value such as $s,t$-shortest path, girth, and
diameter. We also note that the above techniques for proving bounds in
unweighted graphs also apply to the theorem below.

\begin{theorem}\label{thm:st_tradeoffs}
    No algorithm can solve the $s,t$-shortest path, girth (directed), or diameter
    problems in planar graphs on $N$ nodes with amortized update time $u(N)$ and
    query time $q(N)$ such that $\max(q(N),u(N)) =
    O(N^{\frac{1}{2}-\eps})$ for any $\eps > 0$ unless Conjecture~\ref{conj:apsp} is
    false. Furthermore, if $q(N) \ge u(N)$ the algorithm cannot have $q(N)\cdot
    u(N) = O(N^{1-\eps})$. This holds even if the planar embedding of $G$
    never changes.
\end{theorem}
\begin{proof}
    We note that the proof follows the exact same structure as the proof
    of Theorem~\ref{thm:apsp_tradeoff} and only mention the changes needed to be made.

    For $s,t$-shortest path and diameter we add two additional nodes $s,t$ to
    the initial graph $G$ and when performing a query of the distance between
    $a_j$ and $a'_{n^\alpha-j+1}$ we instead insert edges $(s,a_j)$ and
    $(t,a'_{n^\alpha-j+1})$ of sufficiently high weight so that this is the
    longest distance in the graph and query the $s,t$ distance.

    For girth we direct all horizontal edges of $G$ to the right, all vertical
    edges of the left grid down and all vertical edges of the right grid up.
    When doing a query we add the directed edge $(a'_{n^\alpha-j+1},a_j)$ with
    weight $1$, and the length of the shortest cycle then corresponds to the
    shortest path from $a_j$ to $a'_{n^\alpha-j+1}$ plus $1$.
\end{proof}

In the above reductions, the condition $q(N) \ge u(N)$ comes from the fact that
we perform $O(1)$ updates for every query we make and the argument from
Theorem~\ref{thm:apsp_tradeoff} thus breaks down if we try to argue for slower updates
than queries. This makes sense from an upper bound perspective: clearly, any
algorithm with $q(N) \le u(N)$ could simply perform a query for every update,
store the answer, and then provide queries in $O(1)$ time.

\section{Weight updates}
We mentioned in the previous sections that the results hold even if we only
allow weight updates instead of edge insertions/deletions. In
\cite{AbrahamCDGW16} they considered this model in which the algorithm is
supplied with an initial graph $G$ and a promise that for any updated graph
$G'$ we have $d_G(u,v)\le d_{G'}(u,v)\le M\cdot d_G(u,v)$ for all $u,v\in G$
and some parameter $M > 1$. The only operations allowed are weight updates and
queries. We note that all the above results for weighted graphs also hold in
this model.

As a proof sketch, consider the result of Theorem~\ref{thm:apsp_tradeoff}:
The only edges whose weight changes are the ``in-between'' edges $(b_k,b'_k)$
whose weights are always between $X^2\cdot(2(n^\alpha+1)(n^\beta-k)$ and
$X^2\cdot(2(n^\alpha+1)(n^\beta-k))+X$. Similarly, for
$s,t$-shortest path and diameter: Assume that the edge $(s,a_j)$ has weight $y$
when added in the reduction of Theorem~\ref{thm:st_tradeoffs}. We may instead
initialize the graph $G$ with an edge $(s,a_j)$ of weight $y$ for each $1\le
j\le n^\alpha$, increase each edge to have weight $M\cdot y$ and then instead
of adding the edge $(s,a_j)$ we decrease its weight back to $y$. We do the
same for $t$ and the nodes $a'_j$. By picking $y$ sufficiently large we may
ensure that an edge of weight $M\cdot y$ cannot be on the shortest path from
$s$ to $t$. Furthermore these changes can still be done while maintaining the
graphs as grids.

\section{Worst-case bounds for partially dynamic problems}
Our reductions above work in the fully dynamic setting, where edge insertions
\emph{and} deletions (or weight increments \emph{and} decrements) are allowed. We now show
that, using standard techniques, we can turn these amortized bounds into
worst-case bounds for the same problem in the incremental and decremental (only
insertions/increments or deletions/decrements allowed). We will show the result
for dynamic APSP and note that the method is the same for the other problems.

\begin{corollary}\label{cor:apsp_worstcase}
    No algorithm can solve the incremental or decremental APSP problem for
    planar graphs on $N$ nodes with worst-case query time $q(N)$ and update time
    $u(N)$ such that $q(N)\cdot u(N) = O(N^{1-\eps})$ for any $\eps > 0$ unless
    Conjecture~\ref{conj:apsp} is false.
\end{corollary}
\begin{proof}
    We present the argument for the problem when we are given an initial graph
    $G$ and are only allowed to increase weights on edges. The proof uses the
    same rollback technique employed before in several papers (see e.g.
    \cite{AbboudV14}).

    First we create the same initial graph, $G$, as in the proof of
    Theorem~\ref{thm:apsp_tradeoff}. We set the initial weight of the edges
    $(b_k,b'_k)$ to be $X^2\cdot(2(n^\alpha+1)(n^\beta-k))$. During the phase
    of each row $i$ of $A$ we keep track of all memory changes made by the
    incremental data structure while increasing each edge to have weight
    $X^2\cdot(2(n^\alpha+1)(n^\beta-k)) + A_{i,k}$. We then perform each
    distance query and instead of deleting the incremented edges, we ``roll
    back'' the data structure using the memory changes we kept track of, thus
    restoring $G$ to its initial state. By doing this we solve the
    $(\min,+)$-Matrix-Mult problem in the exact same way as in
    Theorem~\ref{thm:apsp_tradeoff}. However, we cannot ensure any requirement on the
    amortized running time, as the rollback operations may essentially
    ``restore all credit'' to the data structure in the sense of amortized
    analysis. Thus, the time bounds only apply to worst-case running times.
\end{proof}

\paragraph*{Acknowledgements}
The authors would like to acknowledge Shay Mozes, Oren Weimann and Virginia
Vassilevska Williams for helpful comments and discussions.

%

\bibliographystyle{plain}
\bibliography{lowerbounds}

\appendix

\end{document}